\documentclass[pra,a4paper,superscriptaddress,hyperref,twocolumn]{revtex4-1}
\usepackage{amsmath,amssymb,amsthm,dsfont,braket}
\usepackage[british]{babel}
\usepackage[percent]{overpic}

\newcommand{\kket}[1]{|#1\rangle\rangle}
\newcommand{\bbra}[1]{\langle\langle#1|}
\newcommand{\hilb}[1]{\mathcal{#1}}
\newcommand{\setname}[1]{{\sf #1}}

\DeclareMathOperator{\Tr}{Tr}
\DeclareMathOperator{\spn}{span}

\newtheorem{dfn}{Definition}
\newtheorem{lmm}{Lemma}
\newtheorem{prop}{Proposition}
\newtheorem{cor}{Corollary}

\begin{document}
\title{Quantum conditional operations}

\author{Alessandro \surname{Bisio}}

\email{alessandro.bisio@unipv.it}

\affiliation{Quit  group, Dipartimento  di Fisica,  Universit\`a
  degli studi di Pavia, via Bassi 6, 27100 Pavia, Italy}

\affiliation{Istituto Nazionale  di Fisica Nucleare,  Gruppo IV,
  via Bassi 6, 27100 Pavia, Italy}

\author{Michele \surname{Dall'Arno}}

\email{cqtmda@nus.edu.sg}

\affiliation{Centre    for   Quantum    Technologies,   National
  University of Singapore, 3  Science Drive 2, 117543 Singapore,
  Republic of Singapore}

\author{Paolo \surname{Perinotti}}

\email{paolo.perinotti@unipv.it}

\affiliation{Quit  group, Dipartimento  di Fisica,  Universit\`a
  degli studi di Pavia, via Bassi 6, 27100 Pavia, Italy}

\affiliation{Istituto Nazionale  di Fisica Nucleare,  Gruppo IV,
  via Bassi 6, 27100 Pavia, Italy}

\begin{abstract}
  An  essential   element  of   classical  computation   is  the
  ``if-then''  construct,  that accepts  a  control  bit and  an
  arbitrary gate, and provides conditional execution of the gate
  depending on  the value of  the controlling bit. On  the other
  hand, quantum  theory prevents  the existence of  an analogous
  universal construct accepting a control qubit and an arbitrary
  quantum   gate  as   its  input.    Nevertheless,  there   are
  controllable sets of quantum gates  for which such a construct
  exists. Here  we provide a necessary  and sufficient condition
  for a set  of unitary transformations to  be controllable, and
  we give  a complete  characterization of controllable  sets in
  the two  dimensional case. This result  reveals an interesting
  connection  between the  problem  of  controllability and  the
  problem of extracting information from an unknown quantum gate
  while using it.
\end{abstract}

\maketitle

\section{Introduction}

One of  the key  features of any  programming language  is {\em
  conditional statements}, that run  an arbitrary gate depending
on the value of a  controlling variable. The Boolean ``if-then''
construct is fundamental to  break code sequentiality, it allows
the  implementation   of  conditional  loops  and   it  prevents
recursion from being infinite.

Also  in quantum  computation  controlled gates  play a  crucial
role. This  is the case, e.g.~for  the ubiquitous Controlled-NOT
(C-NOT),  which  is  the   pillar  of  most  quantum  algorithms
\cite{NC00}  (a  remarkable  example  is  the  Shor's  algorithm
\cite{Sho94},   that   relies   on   controlled   routines   for
period-finding).  The  crucial difference between  classical and
quantum  controlled  gates is  that  the  latter allow  for  the
control qubit  to be  in a superposition  of states.   This fact
leads to a further, radical difference with respect to classical
computation.  Indeed, while  the requirement that the  gate is a
variable   of---as  opposed   to  being   hard-coded  into---the
``if-then''  construct  can  be  trivially  implemented  in  the
classical world,  it turns out  to be impossible in  the quantum
one,  as  first noticed  by  Kitaev~\cite{Kit95}  more than  two
decades  ago.  As  it  was  recently shown~\cite{TGMV13,  Soe13,
  AFCB14},  quantum  theory  prevents the  implementation  of  a
universal  quantum ``if-then''  construct, namely  one that  can
control the entire  set of unitary gates, that  is, no black-box
transformation  can map  $U$ to  controlled-$U$ for  any unitary
$U$.   This is  a serious  limitation,  as it  implies that  one
should  provide a  different implementation  of the  ``if-then''
construct for each controlled gate in the algorithm.

On the other  hand, there exist sets  of unitary transformations
for which  the implementation of  a control is  possible. Simple
examples  are: i)  the sets  of jointly  perfectly discriminable
unitaries;  ii)   the  sets  such   that,  for  a   given  state
$\ket{\psi}$, one has $U_i\ket{\psi} \propto U_j\ket{\psi} $ for
all the unitaries in the  set. The latter example corresponds to
the   setting    of   Refs.~\cite{ZRKZPLO11,    ZKRO13},   where
\unexpanded{$\ket{\psi}$} is trivially the  vacuum state, and of
Refs.~\cite{FDDB14, FMKB15},  and shares similarities  with that
of Ref.~\cite{OGHW15}. It is also straightforward to notice that
any set  made of two unitary  transformations \unexpanded{$\{U ,
  V\}$}  is always  controllable, since  if \unexpanded{$|  \psi
  \rangle$}  is an  eigenstate of  \unexpanded{$U^\dag V$}  then
\unexpanded{$U | \psi \rangle \propto V | \psi \rangle $}.

Despite its  relevance and the  recent interest in  the problem,
the pivotal question remains unanswered:  {\em what are the sets
  of gates that quantum theory allows to be controlled through a
  conditional statement?}

Here,  we answer  this  question by  providing  a necessary  and
sufficient  condition  for   a  set  of  quantum   gates  to  be
controllable by a quantum ``if-then'' clause.  Surprisingly, our
result unveils a connection  between the controllability problem
and the task  of extracting information from  an unknown quantum
gate while using it.

\section{Formalization}

In  operational   terms,  a  general  reversible   quantum  gate
$\mathcal{U}$  consists of  a black  box transforming  its input
(left  wire)  into  its  output (right  wire).   The  action  of
$\mathcal{U}$  on  state  $\rho$  is represented  by  a  unitary
operator  $U$, i.e.   $\mathcal{U}(\rho)  =  U \rho  U^\dagger$,
$U^\dag U=UU^\dag  =I$ (where $I$ denotes  the identity matrix).
We say that the unitary  operator $U$ is a \emph{representative}
of the gate $\mathcal{U}$.  Clearly, a gate $\mathcal{U}$ admits
several  representatative unitaries,  differing by  a physically
irrelevant {\em global} phase $e^{i\phi}$.   In this work we use
the  term   \emph{unitary}  when   referring  to   a  particular
representative of a \emph{gate}.

For some particular choice of representative $U$, the controlled
gate C-$\mathcal{U}$ acts on the system as the identity operator
$I$  if the  control qubit  is initialized  in state  $\ket{0}$,
while  it  performs  $U$  if  the  control  qubit  is  in  state
$\ket{1}$, namely
\begin{align}
  \label{eq:cunitary}
  \begin{aligned}
    \begin{overpic}{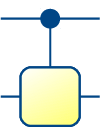}
      \put (30, 16) {$U$}
    \end{overpic}
  \end{aligned}
  = I\otimes \ket{0}\bra{0} +  U\otimes \ket{1}\bra{1}= U \oplus
  I.
\end{align}
Since a phase in front of $U$ in Eq.~\eqref{eq:cunitary} is {\em
  local}   rather  than   global,  different   choices  of   the
representative   for  the   same   gate  $\mathcal{U}$   clearly
correspond to physically inequivalent controlled gates.

To  address  the problem  in  the  most  general case,  we  must
consider a  generic map  that transforms the  gate $\mathcal{U}$
into its controlled version C-$\mathcal{U}$.  From the theory of
\emph{quantum  combs}~\cite{CDP08,  CDP09,  BCDP11} it  is  well
known that  the most  general transformation allowed  by quantum
theory on a quantum gate  $\mathcal{U}$ is realized by inserting
$\mathcal{U}$ in a quantum circuit  board.  Explicitly, if a map
which transforms  $\mathcal{U}$ into C-$\mathcal{U}$  exists, it
corresponds to the following  circuit (further details are given
in the Methods section):
\begin{align}
  \label{eq:controllable}
  \begin{aligned}
    \begin{overpic}{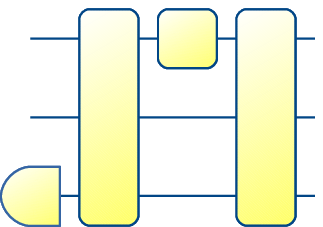}
      \put (7, 6.2) {$0$}
      \put (30,30) {$A$}
      \put (56,56) {$U$}
      \put (80,30) {$B$}
    \end{overpic}
  \end{aligned}
  \; = \;
  \begin{aligned}
    \begin{overpic}{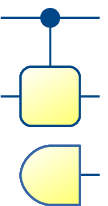}
      \put (19,49) {$U$}
      \put (16,12) {$\psi_U$}
    \end{overpic}
  \end{aligned}\;,
\end{align}
in formula $B (U \otimes I) A  (I \otimes \ket{0}) = (U \oplus I
) \otimes \ket{\psi_U}$, where $0$ denotes the preparation of an
ancillary  ready state  $\ket{0}$,  $A$ and  $B$ denote  unitary
transformations, $U$  is a representative of  $\mathcal{U}$, and
$\ket{\psi_U}$ is a state that  depends on $U$.  It is important
to remark that the dimension of the Hilbert space in the circuit
in  Eq.~\eqref{eq:controllable}  is always  bounded~\cite{CDP09,
  BDPC11}.

The freedom in the choice of the representative, leading to many
inequivalent controlled gates, requires to split the formulation
of the controllability problem  into two sub-problems. The first
one regards controllability of a set of representatives
\begin{dfn}
  Given a set $\setname{S}$ of  gates and a set $\setname{R}$ of
  representatives of $\setname{S}$, we say that $\setname{S}$ is
  \emph{controllable  with   representatives  $\setname{R}$}  if
  there exists $A$ and $B$ such that Eq.~\eqref{eq:controllable}
  holds for any $U \in \setname{R}$.
\end{dfn}
The  second question  regards the  existence  of such  a set  of
representatives, as follows.
\begin{dfn}
  We say that a set $\setname{S}$ of gates is {\em controllable}
  if  there exists  a  set $\setname{R}$  of representatives  of
  $\setname{S}$  such that  $\setname{S}$ is  \emph{controllable
    with representatives $\setname{R}$}.
\end{dfn}

\section{Information without disturbance}

Equation~\eqref{eq:controllable} makes  it manifest  that, along
with the desired task of  controlling $U$, some side information
about  the unknown  unitary $U$  can in  principle be  stored in
state $\ket{\psi_U}$.  The study  of this side information plays
a fundamental role  in our analysis of  controllability. To make
this explicit, consider  the case in which the  control state in
Eq.~\eqref{eq:controllable} is set to $\ket{1}$. We have then
\begin{align}
  \label{eq:markable}
  \begin{aligned}
    \begin{overpic}{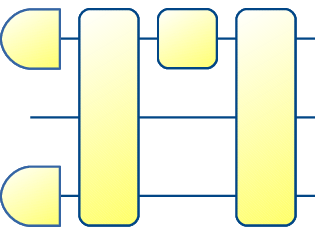}
      \put (7, 6.2) {$0$}
      \put (7, 56) {$1$}
      \put (30,31) {$A$}
      \put (55,55) {$U$}
      \put (80,31) {$B$}
    \end{overpic}
  \end{aligned}
  \; = \;
  \begin{aligned}
    \begin{overpic}{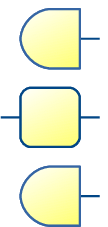}
      \put (20,80) {$1$}
      \put (17,46) {$U$}
      \put (15,12) {$\psi_U$}
    \end{overpic}
  \end{aligned}\;.
\end{align}
Equation~\eqref{eq:markable}    is    an   instance    of    the
information-disturbance  trade-off   problem  in   estimating  a
quantum transformation  \cite{BCDP10}.  Let  us suppose  that we
are provided  with a black box  implementing a single use  of an
unknown transformation  $U$ belonging  to a given  set $\setname
R$.   On   one  hand,   one  wants   to  identify   the  unknown
transformation $U$, while on the other hand one is interested in
applying the  black box on  a variable input state.   In general
these  two tasks  are  incompatible, and  there  is a  trade-off
between the amount  of information that can be  obtained about a
black  box and  the disturbance  caused on  its action,  with an
exception:  when the  black boxes  can be  jointly discriminated
without  error   \cite{Aci01,  DLP02,  DFY07},  and   then  also
reproduced.   The  circuit  in Eq.~\eqref{eq:markable}  fixes  a
scenario in which the unknown unitary must be left unperturbed.

Although  we defer  the details  to the  Methods section,  it is
important to notice here that without loss of generality one can
take the states  $\ket{\psi_U}$ and $\ket{\psi_V}$ corresponding
to unitaries $U$ and $V$ in Eq.~\eqref{eq:markable} to be either
proportional  or  orthogonal.   Indeed, since  the  linear  span
$\setname U$ of maps $\mathcal U$ is a finite dimensional space,
and the circuit in Eq.~\eqref{eq:markable} acts linearly on maps
$\mathcal U$,  also the  span of the  states $\psi_U$  is finite
dimensional,  with a  dimension  bounded by  $\dim \setname  U$.
Now, for  every pair  of unitaries  $U$ and  $V$, the  amount of
information provided  by the circuit  in Eq.~\eqref{eq:markable}
is a decreasing  function of $|\braket{\psi_U|\psi_V}|$. Suppose
that  for  a given  pair  $U,V$  the circuit  providing  maximum
information       about      the       pair      $U,V$       has
$0<\alpha:=|\braket{\psi_U|\psi_V}|<1$.   Then  by applying  the
same    circuit   twice    one    has   $\ket{\psi_U^{(2)}}    =
\ket{\psi_U}\otimes\ket{\psi_U}$,                            and
$|\braket{\psi_U^{(2)}|\psi_V^{(2)}}| =  \alpha^2<\alpha$, while
using  an ancillary  system  having the  same  dimension as  the
initial one. This  implies that the hypothesis  of optimality of
$\alpha$ is  absurd.  One  must then  have either  $\alpha=0$ or
$\alpha=1$. The argument can be repeated for every pair $U,V$.

This observation motivates the following definition.
\begin{dfn}[Markable set of unitaries]
  \label{def:markableset}
  Let $\setname{R}$ be a set  of unitaries, and let $\setname{P}
  := \{  \setname{P}_n \}$ be  a partition of  $\setname{R}$. We
  say  that $\setname  R$  is \emph{$\setname{P}$-markable},  if
  there exist unitaries $C$ and $D$  such that, for any $n$, any
  $U \in  \setname{P}_n$, and  some orthonormal set  $\{ \ket{n}
  \}$ one has
  \begin{align}\label{eq:markableset}
    \begin{aligned}
      \begin{overpic}{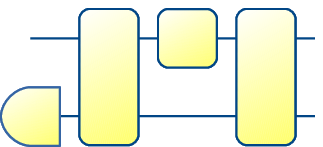}
        \put (7,6) {$0$}
        \put (31,19) {$C$}
        \put (56,30) {$U$}
        \put (80,19) {$D$}
      \end{overpic}
    \end{aligned}
    \; = \;
    \begin{aligned}
      \begin{overpic}{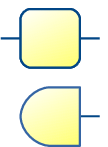}
        \put (28,67) {$U$}
        \put (30,14) {$n$}
      \end{overpic}
    \end{aligned}\;.
  \end{align}
  Similarly,  we  say  that  a set  $\setname{S}$  of  gates  is
  $\setname{P}$-markable if there exists  a set $\setname{R}$ of
  representatives      of       $\setname{S}$      which      is
  $\setname{P}$-markable.
\end{dfn}

For the  sake of  precision, we made  a distinction  between the
notions  of a  markable  set  of gates  and  a  markable set  of
unitaries.   However, this  distinction is  not substantial:  if
$\setname{R}$ is  a set of representatives  of $\setname{S}$ and
$\setname{R}$  is  $\setname{P}$-markable,  then any  other  set
$\setname{R}'$   of   representatives    of   $\setname{S}$   is
$\setname{P}$-markable.

It  is  worth  making  some   easy  considerations:  i)  a  {\em
  necessary}     condition    for     $\setname{R}$    to     be
$\setname{P}$-markable  is that  any  $U  \in \setname{P}_n$  is
perfectly discriminable from  any $V \in \setname{P}_m  , n \neq
m$, and ii) a {\em sufficient} condition for $\setname{R}$ to be
$\setname{P}$-markable   is  that   $\setname{S}$  is   made  of
unitaries that are jointly perfectly discriminable (in this case
$\setname{S}$  is   $\setname{P}$-markable  for   any  partition
$\setname{P}$). As we will prove later, none of these conditions
is  both necessary  and sufficient.

As proved in the Methods section, another simple, yet important,
property is the following one.

\begin{lmm}[Uniqueness of the minimal markable partition]
  \label{lmm:minpart}
  For any set $\setname{R}$ of  unitaries, there exists a unique
  \emph{minimal} partition $\setname{P}$ such that $\setname{R}$
  is   $\setname{P}$-markable    and   $\setname{R}$    is   not
  $\setname{P}'$-markable for  any refinement  $\setname{P}'$ of
  $\setname{P}$.
\end{lmm}

A  relevant feature  of the  information-disturbance problem  in
Eq.~\eqref{eq:markable}  is that  information about  the unknown
$U$ is available  only \emph{after} $U$ has been  applied to the
input state. A more restrictive scenario is the one in which the
outcome of  the estimation  is available \emph{before}  we apply
the unitary and it is described by the circuit
\begin{align}
  \label{eq:marklearn}
  \begin{aligned}
      \begin{overpic}{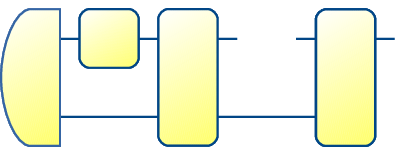}
        \put (7, 16) {$\chi$}
        \put (24, 24) {$U$}
        \put (45, 16) {$C$}
        \put (85, 16) {$D$}
      \end{overpic}
  \end{aligned}
  \; = \;
  \begin{aligned}
      \begin{overpic}{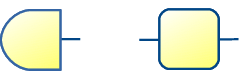}
        \put (10,9) {$n$}
        \put (74,9) {$U$}
      \end{overpic}
  \end{aligned}\;.
\end{align}
where the index  $n$ labels the element of a  partition of a set
of    unitaries     and    $\braket{n|m}     =    \delta_{n,m}$.
Eq.~\eqref{eq:markableset}   and~\eqref{eq:marklearn}  represent
two inequivalent conditions, since in  the second case one could
choose the input  state of $U$ depending on  state $\ket{n}$. An
estimation  without disturbance  with  the  second procedure  is
possible if and only if all the unitaries in the set are jointly
perfectly discriminable.

This result  is equivalent to the  following statement: \emph{if
  \unexpanded{$\{\setname{P}_n\}$} is the minimal partition of a
  set   of   unitaries  \unexpanded{$\setname{R}$}   such   that
  Eq.~\unexpanded{\eqref{eq:marklearn}}  holds,   then  all  the
  unitaries   in  a   subset  \unexpanded{$\setname{P}_n$}   are
  proportional (i.e.  they must  represent the same gate)}.  Let
us   suppose   that   \unexpanded{$\setname{P}_{n'}$}   contains
\unexpanded{$k$}    unitaries   \unexpanded{$\{    U_{i}   \}$},
\unexpanded{$i=1,\dots,  k$}.   The \unexpanded{$\{  U_{i}  \}$}
cannot   be  jointly   perfectly  discriminable,   otherwise  an
iteration   of  the   procedure  would   refine  \unexpanded{$\{
  \setname{P}_{n} \}$} which is minimal by hypothesis.  Then, if
\unexpanded{$\ket{\chi}$}       is       the      state       in
Eq.~\unexpanded{\eqref{eq:marklearn}}    (without     loss    of
generality, \unexpanded{$\ket{\chi}$} can be  assumed to be pure
\cite{CDP09}),   there  must   exist  \unexpanded{$U_i,U_j   \in
  \setname{P}_{n'}$}  such that  \unexpanded{$\bra{\chi}U_i^\dag
  U_j       \otimes      I       \ket{\chi}      \neq       0$}.
Eq.~\unexpanded{\eqref{eq:marklearn}}    implies   \unexpanded{$
  (I\otimes D) (C \otimes  I)(I \otimes U_i\otimes I) \ket{\chi}
  \ket{\psi}       =        \ket{n}\otimes       U_i\ket{\psi}$}
\unexpanded{$\forall  \ket{\psi}$},  and  by taking  the  scalar
product   with  \unexpanded{$i   \neq  j$}   we  easily   obtain
\unexpanded{$\bra{\chi} (U_i^\dagger U_j\otimes  I) \ket{\chi} =
  \bra{\psi}U_i^\dagger   U_j\ket{\psi}$}   \unexpanded{$\forall
  \ket{\psi}$}.    From    \unexpanded{$\bra{\chi}U_i^\dag   U_j
  \otimes  I  \ket{\chi}  \neq  0$},  we  have  \unexpanded{$U_i
  \propto U_j$}.

The problem of deriving the minimal partition $\setname{P}$ such
that    a   given    set   $\setname{R}$    of   unitaries    is
$\setname{P}$-markable is difficult in general.  Our main result
on  the   markability  of   unitaries  is  the   following  full
characterization of the  sets of markable unitaries  of a qubit.

\begin{prop}[Markable sets of qubit unitaries]
  \label{prop:qubitmarkableset}
  A     set    $\setname{R}$     of    qubit     unitaries    is
  $\setname{P}$-markable   with  respect   to  the   non-trivial
  bipartition $\setname{P}  := \{\setname{P}_0, \setname{P}_1\}$
  if and  only if $\spn(\setname{P}_0)  \cap \spn(\setname{P}_1)
  =\{0\}$   and  i)   either   both  $\spn(\setname{P}_0)$   and
  $\spn(\setname{P}_1)$  are  at  most two-dimensional,  or  ii)
  $\setname R$ is jointly discriminable
\end{prop}

While  the  proof  of  necessity  is  rather  technical  and  is
therefore deferred  to the  Methods section,  it is  relevant to
provide  here  a  constructive proof  of  sufficiency.   Suppose
without loss  of generality that  $\spn(\setname{P}_0) \subseteq
\spn(\{I,   \sigma_z\})$   and  $\spn(\setname{P}_1)   \subseteq
\spn(\{\sigma_x, \sigma_y\})$.  By the circuit
\begin{align}
  \label{eq:qubitmarkableset}
  \begin{aligned}
    \begin{overpic}{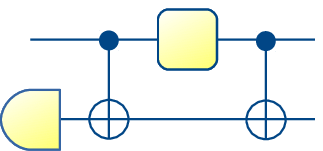}
      \put (8,6) {$0$}
      \put (56,31) {$U$}
    \end{overpic} 
  \end{aligned}
  \; = \;
  \begin{aligned}
    \begin{overpic}{rhs04}
      \put (28,67) {$U$}
      \put (30,14) {$n$}
    \end{overpic}
  \end{aligned}\;,
\end{align}
one   can    then   easily    check   that    $\setname{R}$   is
$\setname{P}$-markable.

Proposition~\ref{prop:qubitmarkableset}   suggests    a   simple
procedure   to  determine   the  existence   of  a   bipartition
$\setname{P}$ such that a set  $\setname{R} $ of qubit unitaries
is  $\setname{P}$-markable: i)  diagonalize  an arbitrary  $U\in
\setname{R}  $,  $U  \not  \propto I$,  ii)  check  whether  the
unitaries in  $\setname{R}$ are either diagonal  or off-diagonal
in the eigenbasis  of $U$, iii) if this is  not the case, repeat
step (ii) for the unitaries  in $U^\dag\setname{R}$.  If the set
is markable, the  minimal partition will be a  refinement of the
partition  $\{\setname{P}_0,\setname{P}_1\}$   corresponding  to
diagonal  and off-diagonal  elements, respectively.   If neither
step (ii)  nor step (iii)  provide a  partition, the set  is not
markable. Further  refinements are possible  if and only  if the
set  $ \setname{R}  $ is  either made  of either  three or  four
jointly discriminable unitaries.  To verify this condition, both
{$\setname      P_0$}      and       the      set      $\setname
P'_1:=U^{(1)\dag}_n\setname  P_1$  of  diagonal  unitaries  must
split  into  a  subset  proportional   to  the  identity  and  a
trace-less  one.   If so,  the  splittings  provide the  minimal
partition,   otherwise    $\{\setname{P}_0,\setname{P}_1\}$   is
minimal.

Proposition~\ref{prop:qubitmarkableset}  implies that:  i) joint
discriminability  is  not  necessary  for  $\setname{R}$  to  be
markable [contrarily  to the case in  Eq. \eqref{eq:marklearn}];
and  ii)   the  existence  of  a   bipartition  $\setname{P}  :=
\{\setname{P}_0,\setname{P}_1\}$  such   that  any   unitary  in
$\setname{P}_0$ is  perfectly discriminable from any  unitary in
$\setname{P}_1$ is not sufficient for $\setname P$-markability.

\section{Controllability}

Thanks to these preliminary considerations,  we are now ready to
state our main result on controllability.

\begin{prop}[Necessary    and     sufficient    condition    for
  controllability]
  \label{prop:control}
  Let  $\setname{R}$   be  a   set  of  unitary   operators  and
  $\setname{P} := \{ \setname{P}_n  \}$ be the minimal partition
  of     $\setname{R}$     such    that     $\setname{R}$     is
  $\setname{P}$-markable.  Then $\setname{R}$ is controllable if
  and only if there exists  a vector $\ket{\psi}$ such that, for
  any $n$ and any $U, V \in \setname{P}_n$, we have
  \begin{align}
    \label{eq:control}
    V^\dagger U \ket{\psi} = \ket{\psi}.
  \end{align}  
\end{prop}

While  the  proof  of  necessity  is  rather  technical  and  is
therefore deferred  to the  Methods section,  it is  relevant to
provide here a  constructive proof of sufficiency.   Let $C$ and
$D$   be   the   unitaries   that   realize   the   circuit   in
Eq.~\eqref{eq:markableset}    for    the    minimal    partition
$\setname{P}:=   \setname{P}_n$  such   that  $\setname{R}$   is
$\setname{P}$-markable.  Then, one can verify that the following
circuit controls the set $\setname{R}$:
\begin{align}
  \label{eq:controlcircuitgen}
  \begin{aligned}
    \begin{overpic}{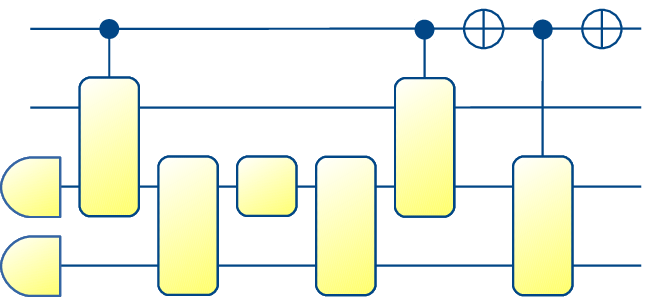}
      \put (4,16) {$\psi$}
      \put (4,3) {$0$}
      \put (15.5,21) {$S$}
      \put (27.5,9) {$C$}
      \put (40,15) {$U$}
      \put (52,9) {$D$}
      \put (64.5,21) {$S$}
      \put (83,9) {$T$}
    \end{overpic}
  \end{aligned}
  \; = \;
  \begin{aligned}
    \begin{overpic}{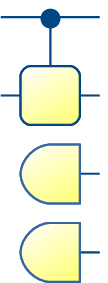}
      \put (15,8) {$n$}
      \put (15,38) {$\psi$}
      \put (13.5,65) {$U$}
    \end{overpic}
  \end{aligned}\;,
\end{align}
where   $S$  is   the  swap   operator  ($S   \ket{a}\ket{b}  :=
\ket{b}\ket{a}$)  and  we  defined $  T  :=  \sum_nV^\dagger_{n}
\otimes  \ket{n}\bra{n}   $  where  $V_n$  is   any  unitary  in
$\setname{P}_n$.     Indeed,    \emph{after}    the    use    of
\unexpanded{$U$}   in   the   circuit,   the   classical   index
\unexpanded{$n$}   (encoded  in   the  lower   output  wire   of
\unexpanded{$D$}) is available.

As proved  in the Methods  section, as a trivial  consequence of
Prop. \ref{prop:control} we have
\begin{cor}
  \label{cor:controlgate}
  Let $\setname{S}$  be a  set of gates  and $\setname{P}  := \{
  \setname{P}_n  \}$   be  the   unique  minimal   partition  of
  $\setname{S}$       such      that       $\setname{S}$      is
  $\setname{P}$-markable.  Then $\setname{S}$ is controllable if
  and  only  if,  for  any  choice  of  the  representative  set
  $\setname{R}$, there  exists a  vector $\ket{\psi}$  such that,
  for any $n$ and any $U, V \in \setname{P}_n$, we have
  \begin{align*}
    V^\dagger U \ket{\psi} \propto \ket{\psi}.
  \end{align*}
\end{cor}

The necessary  and sufficient  condition for  controllability of
Proposition~\ref{prop:control} requires knowledge of the minimal
partition    $\setname{P}$    such   that    $\setname{R}$    is
$\setname{P}$-markable---which is  usually difficult  to obtain.
However,     as    proved     in     the    Methods     section,
Propositions~\ref{prop:qubitmarkableset}  and~\ref{prop:control}
allow   for   the   following   complete   characterization   of
controllable sets of qubit unitaries.

\begin{cor}[Controllable sets of qubit unitaries]
  \label{cor:controlqubit}
  A set $\setname{R}$ of qubit  unitaries is controllable if and
  only if it is non trivially markable or it is commuting.
\end{cor}

When considering controllable sets of qubit gate, the circuit in
Eq.~\eqref{eq:controlcircuitgen} simplifies as follows:
\begin{align}
  \label{eq:controlcircuitqubit}
  \begin{aligned}
    \begin{overpic}{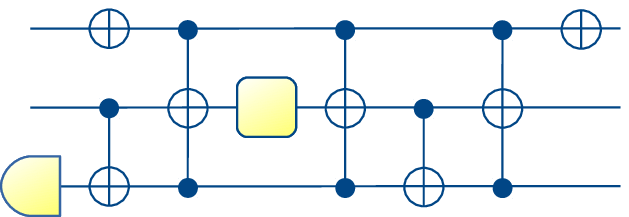}
      \put (4,3.5) {$0$}
      \put (41,16) {$U$}
    \end{overpic}
  \end{aligned}
  \; = \;
  \begin{aligned}
    \begin{overpic}{rhs02}
      \put (19,50) {$U$}
      \put (21,11) {$n$}
    \end{overpic}
  \end{aligned}.
\end{align}
In this case a single ancillary qubit is sufficient.

As a trivial consequence of Corollary~\ref{cor:controlqubit}, we
have  that   the  set   $\setname{R}$  of   all  gates   is  not
controllable, since e.g. the  set proposed in Ref. \cite{AFCB14}
$\setname{R} :=  \{ W_1, W_2, W_3  \}$, with $W_1 =  (\sigma_x +
\sigma_y)/\sqrt{2}$, $W_2 = (\sigma_y + \sigma_z)/\sqrt{2}$, and
$W_3  = (\sigma_z  +  \sigma_x)/\sqrt{2}$ does  not fulfill  the
hypothesis of Corollary~\ref{cor:controlqubit}.

\section{Conclusions}

In this  work we  explored under  what conditions  can a  set of
quantum  conditional statements  be implemented.   We derived  a
necessary and sufficient condition for a set of quantum gates to
be controllable  along with  a complete characterization  of the
controllable  sets of  qubit unitaries.   These results  show an
intimate  relation  between  controllability  and  the  task  of
marking the unitaries, i.e. classifying them while applying them
to an unknown input state.  We completely solved the markability
problem  for two-dimensional  unitaries through  the circuit  of
Eq.~\eqref{eq:qubitmarkableset}, which  could be  considered for
experimental   implementation  using   e.g.~the  technology   of
Ref.~\cite{OPWRB03}. The problem  of finding general markability
conditions in higher dimension remains open.

\section{Acknowledgements}

We thank  M.  Sedlak  for many useful  discussions in  the early
stage of this work.  M.   D.  is supported by Singapore Ministry
of   Education   Academic   Research    Fund   Tier   3   (Grant
No. MOE2012-T3-1-009).

\appendix

\section{Quantum circuit boards}

Let  us   consider  four   finite  dimensional   Hilbert  spaces
$\hilb{H}_j$,  $j=0,  \dots  3$  with dimensions  $d_i  :=  \dim
(\hilb{H}_i)$.   Any  map  that   transforms  an  input  channel
$\mathcal{T}_{in}:\mathcal{B}(\hilb{H}_1)\to
\mathcal{B}(\hilb{H}_2)$     into      an     output     channel
$\mathcal{T}_{out}:\mathcal{B}(\hilb{H}_0)\to
\mathcal{B}(\hilb{H}_3)$ can be realized  by inserting the input
channel into a quantum circuit board as follows:
\begin{align}
  \label{eq:combgen}
  \begin{aligned}
    \begin{overpic}{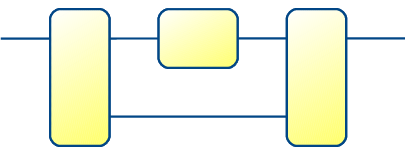}
      \put (4, 30) {$0$}
      \put (31.2, 30) {$1$}
      \put (62, 30) {$2$}
      \put (91, 30) {$3$}
      \put (43.5, 24) {$T_{in}$}
      \put (46, 10) {$A$}
      \put (15.5, 16) {$X_1$}
      \put (72.5, 16) {$X_2$}
    \end{overpic}
  \end{aligned}
  \; = \;
  \begin{aligned}
    \begin{overpic}{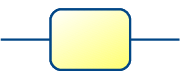}
      \put (10.5,22) {$0$}
      \put (80,22) {$3$}
      \put (33,12) {$T_{out}$}
    \end{overpic}
  \end{aligned}\;,
\end{align}  
where   $\hilb{H}_A$  is   an   ancillary   Hilbert  space   and
$\mathcal{X}_1                       :\mathcal{B}(\hilb{H}_0)\to
\mathcal{B}(\hilb{H}_1 \otimes \hilb{H}_A  )$ and $\mathcal{X}_2
:\mathcal{B}(\hilb{H}_2     \otimes     \hilb{H}_A     )     \to
\mathcal{B}(\hilb{H}_3)$  are  quantum channels.   This  result,
along with many  other properties of quantum  circuit boards, is
well  known  and  is  the  subject  of  many  publications  (see
e.g. \cite{CDP08, CDP09, BDPC11}).

In particular it is known that  any quantum circuit board, as in
Eq.~\eqref{eq:combgen},  corresponds  to   a  positive  operator
(called  \emph{quantum  comb})   $R  \in  \mathcal{B}(\hilb{H}_0
\otimes  \hilb{H}_1  \otimes  \hilb{H}_2 \otimes  \hilb{H}_3)  $
subject to linear constraints.  Obviously, as far as the Hilbert
spaces $\hilb{H}_i$ are  finite dimensional, the set  of all the
admissible quantum circuit board is a compact set.

Also,  the quantum  channels which  realize the  quantum circuit
board  can be  dilated to  unitary channels  acting on  a larger
Hilbert space $\hilb{H}_B$ whose  dimension $d_B$ satisfies $d_B
\leq d_0 d_1 d_2 d_3$.  Then,  for each quantum circuit board as
in  Eq.~\eqref{eq:combgen}, there  exist  two unitary  operators
$A_1,A_2 \in \mathcal{B}(\hilb{H}_B)$,  $\hilb{H}_B = \hilb{H}_0
\otimes \hilb{H}_1 \otimes \hilb{H}_2  \otimes \hilb{H}_3 $ such
that
\begin{align}
  \begin{aligned}
    \begin{overpic}{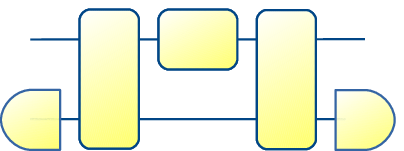}
      \put (6,5) {$0$}
      \put (23,17) {$A_1$}
      \put (44.5,25) {$T_{in}$}
      \put (68,17) {$B_1$}
      \put (88,5) {$I$}
    \end{overpic}
  \end{aligned}
  \; = \;
  \begin{aligned}
    \begin{overpic}{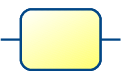}
      \put (26,20) {$T_{out}$}
    \end{overpic}
  \end{aligned}\;,
\end{align}  
where  $\ket{0}   \in  \hilb{H}_1  \otimes   \hilb{H}_2  \otimes
\hilb{H}_3 $  and $I$ denotes  the trace on  $\hilb{H}_0 \otimes
\hilb{H}_1 \otimes \hilb{H}_2$.

\section{Information without disturbance}

Let $ \{ U_i\}$ be a  (possibly infinite) set of $SU(d)$ unitary
operators.   Each  of  them  corresponds to  a  unitary  channel
$\mathcal{U}_i:                       \mathcal{B}(\hilb{H}_1)\to
\mathcal{B}(\hilb{H}_2)$,   $\mathcal{U}_i(\rho)   =  U_i   \rho
U^\dag_i$,  $d_1=d_2 =d$.   Let  us consider  a quantum  circuit
board such that
\begin{align}
  \label{eq:markcomb}
  \begin{aligned}
    \begin{overpic}{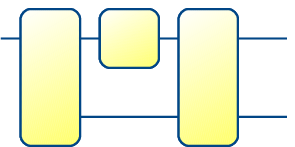}
      \put (12,20) {$Y_1$}
      \put (67,20) {$Y_2$}
      \put (40,33) {$U_i$}
      \put (86,13) {$E$}
    \end{overpic}
  \end{aligned}
  \; = \;
  \begin{aligned}
    \begin{overpic}{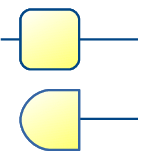}
      \put (24.5,74) {$U_i$}
      \put (29,16) {$\psi_i$}
      \put (70,27) {$E$}
    \end{overpic}
  \end{aligned}\;,
\end{align}  
for all  $i$ and for  some set  of pure states  $\{ \ket{\psi_i}
\}$.  Whichever the  set $\{ \ket{\psi_i} \}$  is, the linearity
of  the circuit  in Eq.~\eqref{eq:markcomb}  implies that  $\dim
(\spn \{ \ket{\psi_i}\})  \leq \dim (\spn \{ U_i  \}) \leq d^2$.
Then  we can  consider the  quantum circuit  board in  which the
Hilbert  space  $\hilb{H}_E$,  which   carries  the  states  $\{
\ket{\psi_i} \}$,  is encoded into an  $d^2$-dimensional Hilbert
space.  The resulting circuit board correspond to a quantum comb
$R \in \mathcal{B}(\hilb{H}^{\otimes 6})$.

Let us denote with $\setname{M}$  the set of the quantum circuit
boards which obey Eq.  \eqref{eq:markcomb}  for some set of pure
states. We have the following result.

\begin{lmm}
  \label{lem:compact}
  The set $\setname{M}$ is compact.
\end{lmm}

\begin{proof}
  The set $\setname{M}$ corresponds to the set defined as 
  \begin{align}
    \label{eq:compact}
    \begin{cases}
      R \in \mathcal{B}(\hilb{H}^{\otimes 6})
      \textrm{ is a quantum comb},\\
      \Tr_E[\bbra{U_i^*}R\kket{U_i^*}] = \kket{U_i}\bbra{U_i}, \\
      (\Tr_{0,3}[\bbra{U_i^*}R\kket{U_i^*}])^2                 =
      d^2\Tr_{0,3}[\bbra{U_i^*}R\kket{U_i^*}],
    \end{cases}
  \end{align}
  where    the    last     two    equalities    translate    Eq.
  \eqref{eq:markcomb}   in    terms   of   the    operator   $R$
  ($\kket{U_i}\bbra{U_i}$ is  the Choi  operator of  the unitary
  channel      $\mathcal{U}_i$,      with      the      notation
  $\kket{A}:=\sum_{m,n}a_{m,n}\ket{m}\ket{n}$ for an operator $A
  := \sum_{m,n}a_{m,n} \ket{m}\bra{n}$).  Eq.~\eqref{eq:compact}
  defines  a  closed  subset  of  the  compact  set  $\{  R  \in
  \mathcal{B}(\hilb{H}^{\otimes 6}) \mbox{ is a quantum comb}\}$
  and hence it defines a compact set.
\end{proof}

We also have the following result.

\begin{lmm}
  \label{lem:nestedcombs}
  Let  $R   \in  \setname{M}$  be  a   quantum  circuit  obeying
  Eq.   \eqref{eq:markcomb}  and   let   $R^{(2)}$  denote   the
  application of $R$ twice, i.e.
  \begin{align}
    \label{eq:nestedcombs}
    \begin{aligned}
      \begin{overpic}{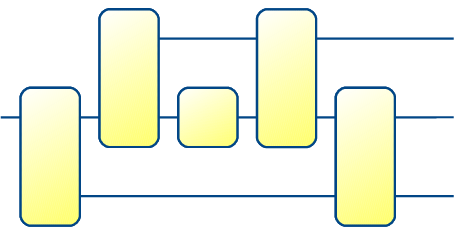}
        \put (8,12) {$Y_1$}
        \put (25,30) {$Y_1$}
        \put (59.5,30) {$Y_2$}
        \put (77,12) {$Y_2$}
        \put (42,22) {$U_i$}
        \put (90,8) {$E$}
        \put (90,43) {$E$}
      \end{overpic}
    \end{aligned}
    \; = \;
    \begin{aligned}
      \begin{overpic}{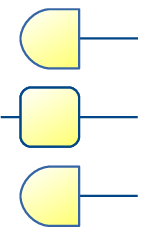}
        \put (17,11) {$\psi_i$}
        \put (17,82) {$\psi_i$}
        \put (15,45) {$U_i$}
        \put (42,17) {$E$}
        \put (42,89) {$E$}
      \end{overpic}
    \end{aligned}\;.
  \end{align}  
  We have that $R^{(2)}\in \setname{M}$.
\end{lmm}

\begin{proof}
  Clearly we  have $\dim  (\spn \{  \ket{\psi_i}\})= \dim(\spn\{
  \ket{\psi_i}\ket{\psi_i}  \})$.   We  can encode  the  Hilbert
  space  $\hilb{H}_E  \otimes  \hilb{H}_E $  which  carries  the
  states    $\{    \ket{\psi_i}\ket{\psi_i}     \}$    into    a
  $d^2$-dimensional Hilbert  space.  Then the  resulting circuit
  board corresponds to a quantum comb in $ \setname{M}$.
\end{proof}

The  quantum  circuit  board  in  Eq.~\eqref{eq:combgen},  while
leaving unaffected the  transformation $\mathcal{U}_i$, extracts
some side  information about $\mathcal{U}_i$ which  is stored in
the state $\ket{\psi_i}$.  Now, for  every pair of unitaries $U$
and $V$,  the amount of  information provided by the  circuit in
Eq.~\eqref{eq:combgen}   can   be   defined  as   an   arbitrary
non-negative       decreasing        function       $f$       of
$\alpha_{U,V}:=|\braket{\psi_U|\psi_V}|$ for any pair $U,V$.

Let now $R^{opt} \in \setname{M}$ be the optimal quantum circuit
board which achieves the maximum value of $f$.  Such a $R^{opt}$
must  exists  since  $\setname{M}$   is  compact  as  proved  in
Lemma~\ref{lem:compact}.  Suppose  that for a given  pair $U,V$,
$R^{opt}$               is               such               that
$0<\alpha_{U,V}:=|\braket{\psi_U|\psi_V}|<1$.    Let    us   now
consider $R^{opt (2)}$ which is the circuit which corresponds to
the   application    of   $R^{opt    }$   twice   as    in   Eq.
\eqref{eq:nestedcombs}.        As       we       proved       in
Lemma~\ref{lem:nestedcombs},  $ R^{opt  (2)}$ is  an element  of
$\setname{M}$.  The quantum circuit  board $ R^{opt (2)}$ gives,
for      any     pair      $U,V$,     $\ket{\psi_U^{(2)}}      =
\ket{\psi_U}\otimes\ket{\psi_U}$                             and
$|\braket{\psi_U^{(2)}|\psi_V^{(2)}}|                          =
\alpha_{U,V}^2<\alpha_{U,V}$.  This implies  that the hypothesis
of optimality of $R^{opt}$ is  absurd. One must then have either
$\alpha_{U,V}=0$ or $\alpha_{U,V}=1$.  Since the argument can be
repeated for every pair $U,V$,  one has that the optimal circuit
board $R^{opt}$  exists and is such  that $\alpha_{U,V}=0,1$ for
all $U,V$. As discussed in the previous section $R^{opt}$ can be
obtained by  a pair of unitary  operators $A^{opt}_1,A^{opt}_2$,
i.e.
\begin{align*}
  \begin{aligned}
    \begin{overpic}{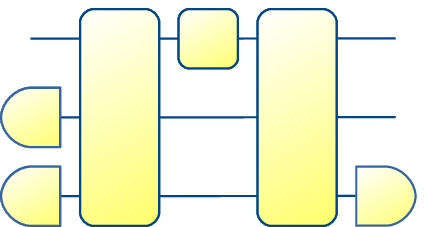}
      \put (7,5) {$0$}
      \put (7,24) {$0$}
      \put (20.8,24) {$A_1^{opt}$}
      \put (64,24) {$A_2^{opt}$}
      \put (46,43) {$U_i$}
      \put (89,5) {$I$}
    \end{overpic}
  \end{aligned}
  \; = \;
  \begin{aligned}
    \begin{overpic}{rhs04}
      \put (24,71) {$U_i$}
      \put (26,15) {$\psi_i$}
    \end{overpic}
  \end{aligned}\;,
\end{align*}
where $|\braket{\psi_i|\psi_j}|  = 0,  1$ for  any $i$  and $j$,
which trivially implies also
\begin{align*}
  \begin{aligned}
    \begin{overpic}{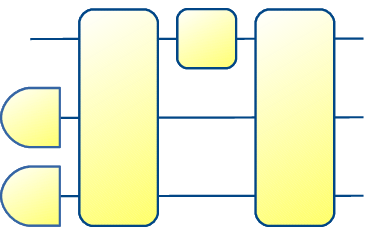}
      \put (7,5) {$0$}
      \put (7,28) {$0$}
      \put (24,28) {$A_1^{opt}$}
      \put (73,28) {$A_2^{opt}$}
      \put (53,49) {$U_i$}
    \end{overpic}
  \end{aligned}
  \; = \;
  \begin{aligned}
    \begin{overpic}{rhs04}
      \put (24,71) {$U_i$}
      \put (26,15) {$\phi_i$}
    \end{overpic}
  \end{aligned}\;,
\end{align*} 
where $|\braket{\phi_i|\phi_j}| = 0, 1$ for any $i$ and $j$.

We can  now prove  Lemma~\ref{lmm:minpart}, that we  report here
for the reader's convenience.

\begin{lmm}[Uniqueness of the minimal markable partition]
  For any set $\setname{R}$ of  unitaries, there exists a unique
  \emph{minimal} partition $\setname{P}$ such that $\setname{R}$
  is   $\setname{P}$-markable    and   $\setname{R}$    is   not
  $\setname{P}'$-markable for  any refinement  $\setname{P}'$ of
  $\setname{P}$.
\end{lmm}

\begin{proof}
  The existence is proved  by considering the trivial partition.
  To    prove    uniqueness,   let    $\setname{P}^{(0)}$    and
  $\setname{P}^{(1)}$ be  two different minimal  partitions.  By
  subsequently applying  the circuit  in Eq.~\eqref{eq:markable}
  for  $\setname{P}^{(0)}$ and  $\setname{P}^{(1)}$, one  proves
  that    $\setname{R}$    is   $\setname{P}'$-markable,    with
  $\setname{P}'$  the  refinement   of  $\setname{P}^{(0)}$  and
  $\setname{P}^{(1)}$.
\end{proof}

We are  now in  a position  to prove our  main result  about the
markability         of        unitaries         (given        in
Proposition~\ref{prop:qubitmarkableset}  and  reported here  for
the reader's  convenience), that is a  complete characterization
of the set of markable unitaries of a qubit.

\begin{prop}[Markable sets of qubit unitaries]
  A     set    $\setname{R}$     of    qubit     unitaries    is
  $\setname{P}$-markable   with  respect   to  the   non-trivial
  bipartition $\setname{P}  := \{\setname{P}_0, \setname{P}_1\}$
  if and  only if $\spn(\setname{P}_0)  \cap \spn(\setname{P}_1)
  =\{0\}$   and  i)   either   both  $\spn(\setname{P}_0)$   and
  $\spn(\setname{P}_1)$  are  at  most two-dimensional,  or  ii)
  $\setname R$ is jointly discriminable
\end{prop}

\begin{proof}
  First we  prove necessity.  The  case ii) is trivial.   Let us
  consider  a   bipartition  $\setname{P}:=   \{\setname{P}_0  ,
  \setname{P}_1     \}$    such     that    $\setname{R}$     is
  $\setname{P}$-markable, and  let us denote by  $U_n^{(i)}$ the
  $n$-th element  of $\setname{P}_i$.  The result  relies on the
  fact   that   two   qubit  unitaries   $U,V$   are   perfectly
  discriminable iff $U^\dag V$ is trace-less \cite{DLP02}, which
  by              Eq.~\eqref{eq:markableset}              implies
  $\Tr[U^{(i)\dag}_mU^{(j)}_n]=0$ for  $i\neq j$.   Without loss
  of generality we suppose  $I \in \setname{P}_0$ (otherwise, we
  can consider the set $\setname{R}' := \{ U^{(0)\dag}_j U_n \}$
  which  is  $\setname{P}$-markable  if  and  only  if  the  set
  $\setname{R}$  is).    We  show  that  $\setname{R}$   is  not
  $\setname{P}$-markable  for  any   partition  $\setname{P}  :=
  \{\setname{P}_0,       \setname{P}_1\}$        such       that
  $\dim(\spn(\setname{P}_1))   =  3$.    Since   we  must   have
  $\dim(\spn(\setname{P}_0))  =   1$,  we   have  $\spn(\setname
  P_0)=\spn(I)$                and                $\spn(\setname
  P_0)=\spn(\sigma_x,\sigma_y,\sigma_z)$,  where  $\spn(\setname
  T)$ denotes the  complex span of $\setname  T$, and $\sigma_i$
  are the Pauli  matrices.  First we consider the  case in which
  $\setname{P}$   is   the   minimal   partition.    Then   from
  Eq.~\eqref{eq:markableset}, defining  $T_i:=D (\sigma_i \otimes
  I)  C (I  \otimes  \ket{0})$,  we must  have  $T_0= I  \otimes
  \ket{0}$, and $T_i = \sigma_i\otimes \ket{1} $, from which one
  derives    the     contradiction    $\sigma_y=iT^\dag_x    T_z
  =T^\dag_0T_y= 0$.   Let us  now suppose that  $\setname{P}$ is
  not the minimal partition and all the unitaries in the set are
  not  jointly  perfectly   discriminable.   Then,  the  minimal
  partition     $\setname{P}'$     must     be     such     that
  $\setname{P}':=\{\setname{P}_0, \setname{P}_1,\setname{P}_2\}$
  with  $\dim(\spn(\setname{P}_0)) =\dim(\spn(\setname{P}_1))  =
  1$  and  $\dim(\spn(\setname{P}_2))  = 2$.   Without  loss  of
  generality  we can  suppose  $\spn(\setname{P}_0) =  \spn(I)$,
  $\spn(\setname{P}_1)        =       \spn(\sigma_z)$        and
  $\spn(\setname{P}_2)) = \spn(\sigma_y , \sigma_x)$.  Then from
  Eq.~\eqref{eq:markableset},  we  must   have  $T_0=  I  \otimes
  \ket{0} $, $T_z= \sigma_z \otimes \ket{1} $ and $T_i= \sigma_i
  \otimes  \ket{2} $  for $i  = x,y$  from which  we obtain  the
  contradiction $\sigma_z=iT^\dag_y T_x = T^\dag_0 T_z=0$.
  
  To prove sufficiency, suppose  without loss of generality that
  $\spn(\setname{P}_0)  \subseteq   \spn(\{I,  \sigma_z\})$  and
  $\spn(\setname{P}_1) \subseteq  \spn(\{\sigma_x, \sigma_y\})$.
  By the circuit
  \begin{align*}
    \begin{aligned}
      \begin{overpic}{lhs06}
        \put (8,6) {$0$}
        \put (56,31) {$U$}
      \end{overpic}
    \end{aligned}
    \; = \;
    \begin{aligned}
      \begin{overpic}{rhs04}
        \put (28,67) {$U$}
        \put (30,14) {$n$}
      \end{overpic}
    \end{aligned}\;,
  \end{align*}
  one   can   then   easily    check   that   $\setname{R}$   is
  $\setname{P}$-markable.
\end{proof}

\section{Controllability}

We     are      now     in      a     position      to     prove
Proposition~\ref{prop:control},  that  we  report here  for  the
reader's convenience.

\begin{prop}[Necessary    and     sufficient    condition    for
  controllability]
  Let  $\setname{R}$   be  a   set  of  unitary   operators  and
  $\setname{P} := \{ \setname{P}_n  \}$ be the minimal partition
  of     $\setname{R}$     such    that     $\setname{R}$     is
  $\setname{P}$-markable.  Then $\setname{R}$ is controllable if
  and only if there exists  a vector $\ket{\psi}$ such that, for
  any $n$ and any $U, V \in \setname{P}_n$, we have
  \begin{align}
    \label{eq:controlprop}
    V^\dagger U \ket{\psi} = \ket{\psi}.
  \end{align}  
\end{prop}

\begin{proof}
  First, let us assume that $\setname{R}$ is controllable.  Then
  there exist unitaries $A$ and $B$ such that
  \begin{align}
    \label{eq:controlcondition}
    B  ( I\otimes  U )  A (\ket{0}  \otimes I  ) =  \ket{\phi_n}
    \otimes (I \oplus U)
  \end{align}
  holds.  By the optimality  argument used to justify definition
  \ref{def:markableset},   it  is   not   restrictive  to   take
  $\ket{\phi_n}   =   \ket{n}$   for  some   orthonormal   basis
  $\{\ket{n}\}$. Let us consider a vector $\ket{\chi}$ such that
  $(I      \oplus      U)     \ket{\chi}=\ket{\chi}$.       From
  Eq.~\eqref{eq:controlcondition}  we  have   $  (I  \otimes  U)
  \ket{\Psi}=  \ket{\Phi_n}$ where  we  defined $\ket{\Psi}:=  A
  (\ket{0} \ket{\chi} )$  and $\ket{\Phi_n}:= B^\dagger (\ket{n}
  \ket{\chi})$.  Since the previous identity holds for any $U\in
  \setname{R}$   for  some   $n$,  there   exists  a   partition
  $\setname{P}'  $ of  $\setname R$  such then  $\setname{P}'_n$
  collects    all     the    $U$    that     satisfy    equation
  \eqref{eq:controlcondition}     with     the     same     $n$.
  Equation~\eqref{eq:controlcondition} can then  be rewritten as
  $ (I \otimes V^\dagger U)  \ket{\Psi}= \ket{\Psi}$ for any $U,
  V  \in  \setname{P}_n$.   Let  us now  consider  an  expansion
  $\ket{\Psi}:=  \sum_{a,b} c_{a,b}\ket{a}\ket{b}$  ($\{ \ket{a}
  \}$ and  $\{\ket{b}\}$ are orthonormal  basis) and let  us fix
  $a'$ such that $c_{a',b}\neq 0$  for some $b$.  By multiplying
  both  sides   of  $   (I  \otimes  V^\dagger   U)  \ket{\Psi}=
  \ket{\Psi}$    by     $\bra{a'}\otimes    I$     we    recover
  Eq.~\eqref{eq:controlprop}  with $\ket{\psi}:=(\bra{a'}\otimes
  I)\ket\Psi$.  We then proved  that controllability implies the
  existence   of   a    partition   $\setname{P}'$   such   that
  Eq.~\eqref{eq:controlprop}  holds  in  each  $\setname{P}'_n$.
  The   same  condition   holds   for   the  minimal   partition
  $\setname{P}$ since it is  a refinement of $\setname{P}'$.  We
  then proved necessity.

  The proof  of sufficiency is as  follows.  Let $C$ and  $D$ be
  the unitaries that  realize the circuit \eqref{eq:markableset}
  for the  minimal partition $\setname{P}:=  \setname{P}_n$ such
  that $\setname{R}$  is $\setname{P}$-markable.  Then,  one can
  verify   that  the   following   circuit   controls  the   set
  $\setname{R}$:
  \begin{align*}
    \begin{aligned}
      \begin{overpic}{lhs08}
        \put (4,16) {$\psi$}
        \put (4,3) {$0$}
        \put (15.5,21) {$S$}
        \put (27.5,9) {$C$}
        \put (40,15) {$U$}
        \put (52,9) {$D$}
        \put (64.5,21) {$S$}
        \put (83,9) {$T$}
      \end{overpic}
    \end{aligned}
    \; = \;
    \begin{aligned}
      \begin{overpic}{rhs08}
        \put (15,8) {$n$}
        \put (15,38) {$\psi$}
        \put (13.5,65) {$U$}
      \end{overpic}
    \end{aligned}\;,
  \end{align*}
  where  $S$   is  the   swap  operator  ($S   \ket{a}\ket{b}  =
  \ket{b}\ket{a}$)  and we  defined $  T :=  \sum_nV^\dagger_{n}
  \otimes  \ket{n}\bra{n}  $  where  $V_n$  is  any  unitary  in
  $\setname{P}_n$.  Indeed,  \emph{after} the use of  $U$ in the
  circuit, the classical index $n$  (encoded in the lower output
  wire of $D$) is available.
\end{proof}

\begin{cor}[Necessary    and     sufficient    condition    for
  controllability]
  Let $\setname{S}$  be a  set of gates  and $\setname{P}  := \{
  \setname{P}_n  \}$   be  the   unique  minimal   partition  of
  $\setname{S}$       such      that       $\setname{S}$      is
  $\setname{P}$-markable.  Then $\setname{S}$ is controllable if
  and  only  if,  for  any  choice  of  the  representative  set
  $\setname{R}$ there  exists a  vector $\ket{\psi}$  such that,
  for any $n$ and any $U, V \in \setname{P}_n$, we have
  \begin{align}
    \label{eq:control2}
    V^\dagger U \ket{\psi} \propto \ket{\psi}.
  \end{align}  
\end{cor}

\begin{proof}
  The conditions of Corollary~\ref{cor:controlgate} are met by a
  set of  representatives if  and only  if they  are met  by any
  other set  of representatives. Given a  set of representatives
  and a  vector such  that Eq.~\eqref{eq:control2} holds,  it is
  straightforward  to find  a set  of representatives  such that
  Eq.~\eqref{eq:controlprop} holds.
\end{proof}

\begin{cor}[Controllable sets of qubit unitaries]
  A set $\setname{R}$ of qubit  unitaries is controllable if and
  only if it is non-trivially markable or it is commuting.
\end{cor}

\begin{proof}
  If  the  set  $\setname{R}$  is markable,  the  conditions  of
  Proposition~\ref{prop:control} are  met since the  elements of
  $\setname{P_0}$  commute, while  multiplying two  off diagonal
  matrices in  $\setname{P_1}$ provides a diagonal  one.  If the
  minimal   partition   of   $\setname{R}$  is   trivial,   then
  controllability    is   equivalent    to   commutativity    of
  $\setname{R}$.
\end{proof}


\begin{thebibliography}{}
\bibitem{NC00} M.  A.  Nielsen and  I.  L.  Chuang, {\em Quantum
    computation and  quantum information}  (Cambridge university
  press, 2010).
\bibitem{Sho94}  P. W.   Shor, in  {\em  Proceedings of the 35th  Annual Symposium on Foundations  of Computer
    Science} (IEEE,
  1994) pp.  124-134.
\bibitem{Kit95}  A.  Y.  Kitaev,  Quantum  measurements and  the
  Abelian       stabilizer      problem.        Preprint      at
  $\langle$http://arxiv.org/abs/quant-ph/9511026$\rangle$
  (1995).
\bibitem{TGMV13} J.  Thompson, M. Gu,  K.  Modi, and V.  Vedral,
  {\em   Quantum   Computing    with   black-box   subroutines},
  arXiv:1310.2927.
\bibitem{Soe13} A. Soeda, {\em Limitations on quantum subroutine
    designing due to the linear structure of quantum operators},
  Talk  at  international   conference  on  quantum  information
  (ICQIT) (2013).
\bibitem{AFCB14} M.   Ara\'ujo, A.   Feix, F.  Costa,  and \v{C}
  Brukner,   {\em  Quantum   circuits  cannot   control  unknown
    operations}, New Journal of Physics {\bf 16}, 093026 (2014).
\bibitem{ZRKZPLO11} X.-Q.   Zhou, T.  C.  Ralph,  P.  Kalasuwan,
  M. Zhang, A. Peruzzo, B. P. Lanyon, J. L. O'Brien, {\em Adding
    control  to arbitrary  unknown  quantum operations},  Nature
  Communication {\bf 2}, 413 (2011).
\bibitem{ZKRO13} X.-Q. Zhou, P.  Kalasuwan,  T. C.  Ralph, J. L.
  O'Brien, {\em  Calculating unknown eigenvalues with  a quantum
    algorithm}, Nature Photonics {\bf 7}, 223 (2013).
\bibitem{FDDB14}  N.  Friis,  V.  Dunjko,  W.  Dur,  and H.   J.
  Briegel,  {\em   Implementing  quantum  control   for  unknown
    subroutines}, Phys.  Rev.  A {\bf 89}, 030303 (2014).
\bibitem{FMKB15}  N. Friis,  A. A.  Melnikov, G.  Kirchmair, and
  H.   J.   Briegel,    {\em   Coherent   controlization   using
    superconducting qubits}, Sci. Rep. {\bf 5}, 18036 (2015).
\bibitem{OGHW15}  M. Oszmaniec,  A.  Grudka,  M. Horodecki,  and
  A. W\'ojcik, {\em Creating  a superposition of unknown quantum
    states}, Phys. Rev. Lett. {\bf 116}, 110403 (2016).
\bibitem{CDP08} G. Chiribella, G. M. D'Ariano, and P. Perinotti,
  {\em Quantum  circuit architecture}, Phys.  Rev.   Lett.  {\bf
    101}, 060401 (2008).
\bibitem{CDP09}  G.   Chiribella,  G.    M.   D'Ariano,  and  P.
  Perinotti, {\em  Theoretical framework for  quantum networks},
  Phys.  Rev.  A {\bf 80}, 022339 (2009).
\bibitem{BCDP11} A.  Bisio, G.  Chiribella, G. M.  D'Ariano, and
  P.   Perinotti,  {\em  Quantum networks:  general  theory  and
    applications},  {\em Acta  Physica  Slovaca}  {\bf 61},  273
  (2011).
\bibitem{BDPC11} A.  Bisio, G.  M.  D'Ariano, P.  Perinotti, and
  G. Chiribella, {\em Minimal computational-space implementation
    of multiround  quantum protocols}, Phys.  Rev.   A {\bf 83},
  022325 (2011).
\bibitem{BCDP10} A. Bisio, G. Chiribella, G. M. D'Ariano, and P.
  Perinotti, {\em Information-disturbance tradeoff in estimating
    a unitary  transformation}, Phys.  Rev.  A  {\bf 82}, 062305
  (2010).
\bibitem{Aci01} A.  Ac\'in,  {\em Statistical distinguishability
    between unitary  operations}, Phys.  Rev.  Lett.   {\bf 87},
  177901 (2001).
\bibitem{DLP02} G.  M.  D'Ariano, P.   Lo Presti, and M.  G.  A.
  Paris, {\em Improved discrimination of unitary transformations
    by  entangled  probes}, Journal  of  Optics  B: Quantum  and
  Semiclassical Optics {\bf 4}, S273 (2002).
\bibitem{DFY07} R. Duan, Y. Feng, and M. Ying, {\em Entanglement
    is not necessary for  perfect discrimination between unitary
    operations}, Phys.  Rev.  Lett.  {\bf 98}, 100503 (2007).
\bibitem{OPWRB03} J. L.  O'Brien, G.  J. Pryde, A. G.  White, T.
  C.   Ralph,  and  D.    Branning,  {\em  Demonstration  of  an
    all-optical quantum controlled-NOT  gate}, Nature {\bf 426},
  264 (2003).
\end{thebibliography}
\end{document}